\documentclass[11pt]{article}
\usepackage{fullpage}

\usepackage{amsfonts,amssymb,amsmath,mathtools}
\usepackage{graphicx}
\usepackage[english]{babel}
\usepackage[latin1]{inputenc}
\usepackage{amsthm}
\usepackage{verbatim}
\usepackage{enumerate}
\usepackage{tikz}
\usetikzlibrary{shapes,snakes}
\usetikzlibrary{backgrounds,calc}
\usepackage{float}
\usepackage{color}
\usepackage{todonotes}
\usepackage[ruled,vlined,titlenotnumbered]{algorithm2e}
\usepackage[colorlinks=true,breaklinks=true,bookmarks=true,urlcolor=blue,
citecolor=blue,linkcolor=blue,bookmarksopen=false,draft=false]{hyperref}
\usepackage[nameinlink,capitalize,english]{cleveref}

\newtheorem{theorem}{Theorem}[section]

\newtheorem{corollary}[theorem]{Corollary}

\newtheorem{lemma}[theorem]{Lemma}

\newtheorem{claimex}[theorem]{Claim}
\crefname{claimex}{Claim}{Claims}

\theoremstyle{definition}
\newtheorem{definition}[theorem]{Definition}
\newtheorem{example}[theorem]{Example}

\allowdisplaybreaks

\usepackage{authblk}

\DeclareMathOperator*{\argmin}{arg\,min}

\newcommand*\diff{\mathop{}\!\mathrm{d}} 
\newcommand{\PP}{\mathcal{P}}
\newcommand{\e}{\mathrm{e}} 

\begin{document}

\title{Bicriteria Nash Flows over Time}

\date{}
\author[1]{Tim Oosterwijk}
\author[2]{Daniel Schmand}
\author[3]{Marc Schr\"{o}der}
\affil[1]{Vrije Universiteit Amsterdam, the Netherlands}
\affil[2]{University of Bremen, Germany}
\affil[3]{Maastricht University, the Netherlands}

\maketitle

\begin{abstract}
Flows over time are a natural way to incorporate flow dynamics that arise in various applications such as traffic networks. In this paper we introduce a natural variant of the deterministic fluid queuing model in which users aim to minimize their costs subject to arrival at their destination before a pre-specified deadline. We determine the existence and the structure of Nash flows over time and fully characterize the price of anarchy for this model. The price of anarchy measures the ratio of the quality of the equilibrium and the quality of the optimum flow, where we evaluate the quality using two different natural performance measures: the throughput for a given deadline and the makespan for a given amount of flow. While it turns out that both prices of anarchy can be unbounded in general, we provide tight bounds for the important subclass of parallel path networks. Surprisingly, the two performance measures yield different results here.
\end{abstract}


\section{Introduction}

%

The number of passenger vehicles in Europe and the traffic volume is steadily increasing \cite{eurostat}.
Traffic congestion impose a huge economic loss to the economy. It is estimated that the European Union suffers a cost of almost 100 billion euro due to congestion annually \cite{europeancommission}.
The control of traffic and optimization of traffic systems can lead to a significantly larger efficiency of the road network and is of central and still rising importance.

As such, there has been a huge effort to understand congestion using theoretical models. The dynamic model that gained most attention for modelling traffic is the deterministic fluid queuing model, already introduced by Vickrey~\cite{vickrey1969congestion}. The travel time of a traffic user heavily depends on choices that are being made by others in the network and such choices are made independently and selfishly. Therefore, it is natural to take a game theoretic approach and study Nash flows over time, also called dynamic equilibria, in which no single user has an incentive to unilaterally deviate from their strategy and improve their utility. Koch and Skutella~\cite{koch2011nash} studied Nash flows over time in this model taking into account that each traffic user wants to minimize their travel time and characterized such Nash flows by means of concatenations of static thin flows. The existence of Nash flows over time, even in multicommodity networks, is formally proven by Cominetti et al.~\cite{cominetti2015dynamic}. More structural insights are gained by Sering and Koch~\cite{sering2019nash} for the model with spillbacks, by Correa et al.~\cite{correa2019price} on the price of anarchy, by Cominetti et al.~\cite{cominetti2021longterm} on the long term behavior and by Graf et al.~\cite{graf2020dynamic} for the model with users that decide based on the current network conditions.

On the other hand, modern traffic planners often use complex simulations to predict traffic behavior. However, these models are often mathematically poorly understood. It is, for example, not always possible to compute exact equilibria. Recently, Ziemke et al.~\cite{ziemke2021flows} showed that there is a strong connection between flows induced by an agent-based simulation using the software MATSim~\cite{matsim} and theoretic Nash flows over time. This shows that great progress has been made on theoretical traffic models.

A common drawback of most of the theoretical models is the simplified assumption that road network users aim for minimizing their arrival time. However, in traffic networks in particular, users are not always that single-minded. Route choices of travellers are often based on multiple criteria. Users take other factors into consideration like the quality of the roads or a general preference over different modes of transport. Very often the route choice is based on an externally given deadline for arrival and travellers choose the \emph{nicest} option that guarantees arrival before the deadline. This multi-criteria optimization of routes is also covered in the traffic simulation MATSim~\cite{matsim} and has not been introduced in the deterministic fluid queuing model.

In this paper we extend the state-of-the-art theoretical traffic models with a multi-criteria objective function. We assume that users try to minimize costs subject to arriving at the sink before a given deadline. Here, costs could be thought of as an intrinsic preference a user has regarding the different route choices and queuing dynamics only play a role for the arrival time of a user. This different assumption on the behaviour leads to very different Nash flows compared to the standard model. In a similar direction, \cite{frascaria2020algorithms} recently proposed a model of flows over time that includes scheduling cost if users arrive earlier or later. As they minimize the sum of scheduling costs and travel times, their approach is not a bicriteria model like we consider.

From a game theoretic perspective, we are generally interested in the inefficiency of a Nash flow compared to an optimal flow. We measure this inefficiency using two natural performance measures. On the one hand we consider the \emph{throughput} objective, which for a given deadline measures the amount of flow that reached its destination before the deadline. On the other hand we consider the \emph{makespan} objective, which for a given amount of flow measures the amount of time that is needed for all the flow to reach the destination.
The fact that an {\em earliest arrival flow} exists~\cite{gale59earliestarrivalflow} implies that an optimal solution for one of the objectives corresponds to an optimal solution for the other objective. 
For a given performance measure, the \emph{price of anarchy} ($PoA$) is defined as the worst possible ratio between the quality of an optimal solution (usually denoted by $f^*$) and the quality of a Nash flow (usually denoted by $f$).
Even though optimal solutions of the two objectives are the same, they have different objective values. Thus they induce different notions for the price of anarchy: the throughput-$PoA$ and the makespan-$PoA$.

We give structural insights into Nash flows over time with costs and deadlines and characterize their inefficiency. First, we show that if path costs are different, any Nash flow over time must be layered.
That is, flow is iteratively sent over a single path until the point in time such that the last particle along this path arrives at the sink at the deadline while taking into account particles of all other paths. Second, we prove that such a layered Nash flow over time is guaranteed to exist in general. Third, we give a complete characterization of the inefficiency of Nash flows over time.
We prove the following three results. (i) In series-parallel networks, both prices of anarchy are unbounded. (ii) In parallel path networks the throughput-$PoA$ is at most $2$, or at most $\e/(\e-1)$ if all transit times are 0. (iii) In parallel path networks the makespan-$PoA$ is at most $\e/(\e-1)$, independent of transit time values. All bounds are tight.

Parallel path networks encompass an important subclass of graphs. It is for example known that price of anarchy bounds for regular Nash flows over time are stronger in parallel path networks than in general graphs \cite{correa2019price}. An important application of parallel networks are scheduling problems in which jobs have to be assigned to machines (see e.g. \cite{agt2007,pinedo2012scheduling}). Other models in which parallel networks are studied are cost sharing games \cite{von2013optimal} and Stackelberg routing games \cite{acemoglu2007competition,harks2021stackelberg,harks2019toll,johari2010investment,ozdaglar2008price}.

In the next section we formally introduce the model. \cref{sec:Nash} presents our results regarding the existence and structure of Nash flows over time. We continue in \cref{sec:PoA} with our results concerning the price of anarchy.

\section{The model}\label{sec:Model}
Let $G=(V,E)$ be a directed graph with two special vertices, a source $s\in V$ and a sink $t\in V$. Each edge $e\in E$ has a finite transit time $\tau_e \geq 0$, a capacity $\nu_e > 0$ and a finite cost $c_e \geq 0$. Additionally, we are given a deadline $D$. Flow has to travel from $s$ to $t$ and departs from $s$ at time $\theta$ with an inflow rate denoted by $u(\theta)$. We assume the inflow rate is constant, i.e., $u(\theta) = u$ for all $\theta$ for some finite $u > 0$. We identify a particle by the time at which it leaves the source, i.e., particle $\theta$ leaves $s$ at time $\theta$. In our model, particles would like to arrive at the sink before the deadline $D$. The following edge dynamics describe how the flow travels through the network. 

Every edge $e \in E$ is endowed with a function $f_e^+(\theta)$ describing the inflow rate into $e$ at time $\theta$. If the inflow rate $f_e^+(\theta)$ is larger than the edge capacity $\nu_e$, a queue will build at the tail of the edge at rate $f_e^+(\theta) - \nu_e$. We denote the mass of this queue at time $\theta$ by $z_e(\theta)$. As the queue empties at rate $\nu_e$, the queue will deplete at a rate equal to $f_e^+(\theta) - \nu_e$ if $f_e^+(\theta) < \nu_e$, at most until $z_e = 0$. This implies that the queuing time $q_e(\theta)$ that a particle that enters edge $e$ at time $\theta$ faces is equal to $z_e(\theta)/\nu_e$. Such a particle therefore leaves the edge at time $T_e(\theta) = \theta + z_e(\theta)/\nu_e + \tau_e$. The queue dynamics are thus determined as follows.
\[\frac{\diff z_e(\theta)}{\diff\theta}=\begin{cases}
f_e^+(\theta)-\nu_e &\text{ if }z_e(\theta)>0 \; ,\\
\max\{f_e^+(\theta)-\nu_e,0\} &\text{ if }z_e(\theta)=0 \; .
\end{cases}\]
Furthermore, these dynamics induce the following outflow rate function $f_e^-(\cdot)$.
\[f_e^-(\theta+\tau_e)=\begin{cases}
\nu_e &\text{ if }z_e(\theta)>0 \; ,\\
\min\{f_e^+(\theta),\nu_e\} &\text{ if }z_e(\theta)=0 \; .
\end{cases}\]
We define the cumulative inflow and outflow function of an edge as
\[
F_e^+(\theta) = \int_0^\theta f_e^+(\xi) \diff \xi \;, \quad
F_e^-(\theta) = \int_0^\theta f_e^-(\xi) \diff \xi \; .
\]
Therefore, the queue mass of an edge $e$ at time $\theta$ can alternatively be computed as $z_e(\theta) = F_e^+(\theta) - F_e^-(\theta+\tau_e)$.

%
A \emph{flow over time} is a set of 
edge inflow rates $(f_e^+)_{e\in E}$ 
satisfying the following flow conservation constraints for all vertices $V\setminus\{t\}$ and for 
all $\theta\geq 0$.
\[\sum_{e\in\delta^+(v)}f_e^+(\theta)-\sum_{e \in \delta^-(v)}f_e^-(\theta)=\begin{cases}
u &\text{ if }v=s \; ,\\
0 &\text{ if }v\neq s,t \; ,
\end{cases}\]
where $\delta^+(v)$ and $\delta^-(v)$ are the set of edges leaving and entering into a node $v$, respectively.

When flow particles in a flow over time arrive at an intermediate node $v \neq s,t$, they immediately continue their route along the following edge on their chosen path.

Let $\PP$ denote the set of all $s$,$t$-paths. Usually, for a path $P \in \PP$ we abuse notation and let $P$ refer to the path, the set of vertices on the path, or the set of edges on the path; as long as no confusion shall arise. For a flow over time $f$ and a path $P\in\PP$, we denote by $\ell^P_v(\theta)$ the time at which particle $\theta$ arrives at node $v$ using path $P$. The values $\ell^P_v(\theta)$ can iteratively be computed as $\ell^P_v(\theta) = T_{uv}(\ell^P_u(\theta))$, using $\ell^P_s(\theta) = \theta$ for all $P \in \PP$.

Observe that the particles on a path $P$ adhere to the first-in-first-out property, i.e.,\ $\ell_v^P(\theta_1) < \ell_v^P(\theta_2) \Leftrightarrow \theta_1 < \theta_2$ for all $v \in P$. 
However, note that $\ell^P_v(\theta)$ and $\ell^{P'}_v(\theta)$ can be different for two paths $P$ and $P'$ and that $\ell^P_v(\theta)$ also depends on particles routed through other paths, even if they originated at the source at a time later than $\theta$.

Most of the time, we will actually describe flows as inflow rates into paths rather than into edges. For a path $P \in \PP$ we define $f_P^+(\theta)$ as the rate of flow from $s$ into the first edge of the path. Since a flow over time adheres to the above flow conservation constraints, we can track a particle along its chosen path through the network, and therefore a set of path inflow rates uniquely induces a set of edge inflow rates. Formally, since $\ell^P_v(\theta)$ is an injective function for all $v \in V$ and $P \in \PP$ and for an edge $e = (v,w)\in E$ we can write $f_e^+(\theta) = \sum_{P\in \PP: e \in P} f_P^+\left(\left(\ell^P_v\right)^{-1}(\theta)\right)$.
A flow over time can therefore also be defined by a set of path inflow rates.

Informally, each flow particle wants to minimize its costs subject to arriving at the sink before the deadline $D$. For path $P\in \PP$, we define $c_P=\sum_{e\in P} c_e$. For a given flow over time $f$, the cost of a particle $\theta$ along path $P$ equals 
    \[
    c_P(\theta)= \begin{cases} c_P &\text{ if } \ell^P_t(\theta)< D \; , \\
    \infty &\text{ else} \; .
    \end{cases}
    \]

\begin{definition}A flow over time $f$ is a \emph{Nash flow (over time)} if for 
all $\theta\geq0$,
\[
f_P^+(\theta)>0\Rightarrow P\in \argmin_{P \in \PP}\{c_P(\theta)\} \; .
\]
\end{definition}
Informally, a flow over time is a Nash flow if no particle can deviate to a path with lower cost and still arrive before the deadline. As an immediate consequence, for a given Nash flow over time $f$ the paths with positive inflow at time $\theta$ have the same costs at time $\theta$, which we denote by $c(\theta)$ (independent of the path).


\section{Structure and existence of Nash flows}\label{sec:Nash}

In this section we derive structural results for Nash flows over time that leads us to introduce an important class of flows over time. We prove that within this class of flows over a time, a Nash flow is guaranteed to exist. Throughout this section, we assume without loss of generality that the paths are numbered such that $c_{P_1}\leq c_{P_2}\leq\ldots\leq c_{P_{|\PP|}}$. We start by deriving a set of necessary conditions for Nash flows over time.
\begin{lemma}\label{lem:mono}
For any Nash flow $c(\theta)$ is non-decreasing in $\theta$.
\end{lemma}
\begin{proof}[Proof]
 Let $f$ be a Nash flow over time with $c(\theta_1)> c(\theta_2)$ and $\theta_1<\theta_2$. We derive a contradiction. Since $c(\theta_1)> c(\theta_2)$, we have that $c(\theta_2)<\infty$ and thus $\ell_t^P(\theta_2)<D$ for the chosen path $P$ of $\theta_2$. By the first-in-first-out principle, if $\theta_1$ chooses $P$, then $\ell_t^P(\theta_1)<\ell_t^P(\theta_2)<D$ and thus $c_P(\theta_1) = c_P(\theta_2)< c(\theta_1)$, contradicting the assumption that $f$ is a Nash flow over time.
\end{proof}

\begin{corollary}\label{cor:Phases}
Any Nash flow consists of phases, where in each phase flow is sent only into paths that have the same total cost.
\end{corollary}

Using this result we define a special class of flows, which we name \emph{layered flows (over time)}. We prove that layered Nash flows always exist, and if the costs of all paths are different then for each network all Nash flows are layered flows. Additionally, also if not all paths have different costs, we show for a special class of networks that they are worst case flows among all Nash flows, in the sense that the least amount of flow arrives at the sink before the deadline. As such, they constitute an important class of Nash flows to show 
bounds on the inefficiency.

\begin{definition}\label{def_layeredNash}
Define a \emph{layered flow (over time)} by $0=\theta_0\leq \theta_1\leq \ldots \leq\theta_{|\PP|}\leq D$ such that for all $\theta\in[0,\theta_{|\PP|})$ we have
\[f_{P_i}^+(\theta) = \begin{cases}u(\theta) & \text{for $\theta \in [\theta_{i-1},\theta_i)$\;,}\\ 0 & \text{otherwise\;,}  \end{cases}\]
and $\ell_t^{P_i}(\theta_i)\geq D$, with equality if $\theta_{i-1}<\theta_i$, for all $i=1,\ldots,|\PP|$.
\end{definition}

Note that there are no restrictions on the inflow pattern after time $\theta_{|\PP|}$. Also, notice that if $\theta_{i-1}=\theta_i$ for some $i$, then path $P_i$ receives no flow that arrives before the deadline at the sink. 
Finally, observe that because path costs $c_{P_i}$ are assumed to be non-decreasing in $i$, a layered flow over time is a Nash flow.

We now show that if all paths have different costs, every Nash flow is a layered flow.
For a given flow over time $f$, let $\bar{\theta}_f$ be the last point in time such that all particles leaving the source in $[0,\bar{\theta}_f)$ arrive at the sink before the deadline $D$.
If the flow $f$ is clear 
we omit this dependence.

\begin{lemma}\label{thm:EQisLayered}
Let $c_{P_1}< c_{P_2}<\ldots< c_{P_{|P|}}$. If $f$ is a Nash flow, then $f$ is a layered flow over time.
\end{lemma}
\begin{proof}
 First, note that if all paths have different costs, $\argmin_{P \in \PP}\{c_P(\theta)\mid \ell^P_t(\theta)< D\}$ is either empty or contains a single element. Thus, in a Nash flow over time, the inflow to any path is either $0$ or $u$ for all $\theta\in[0,\bar{\theta})$. Suppose $f$ is not as claimed. Then, there are $\theta_1 < \theta_2 < \theta_3<\bar{\theta}$ with $f_P^+(\theta_1)=u, f_P^+(\theta_2)=0, f_P^+(\theta_3)=u$. This contradicts \cref{lem:mono}.
 \end{proof}

\begin{theorem}\label{thm:exi}
There always exists a layered Nash flow over time.
\end{theorem}
\begin{proof}
The result follows from the existence of a fixed-point according to Brouwer's fixed point theorem of the function $g:[0,D]^{|\PP|}\rightarrow [0,D]^{|\PP|}$ defined as follows. Given a vector $\theta=(\theta_1,\ldots,\theta_{|\PP|})\in [0,D]^{|\PP|}$, define $\bar\theta_0=0$ and $\bar\theta_i=\max_{j=1,\ldots,i}\theta_j$ for all $i=1,\ldots,k$. For a given $i$, consider the flow over time $f$ with non-constant inflow
\[f_{P_i}^+(\theta) = \begin{cases}u & \text{ for $\theta \in [\bar\theta_{i-1},\infty)$\;,}\\ 0 & \text{otherwise\;,}  \end{cases}\quad\text{ and}\quad
 f_{P_j}^+(\theta) = \begin{cases}u & \text{ for $\theta \in [\bar\theta_{j-1},\bar\theta_j)$\;,}\\ 0 & \text{otherwise\;,}  \end{cases}\] for all $j \neq i$. Define 
 \[g_i(\theta) =\begin{cases}
 \sup\{\theta \mid  \ell^{P_i}_t(\theta)< D\} &\text{ if } \ell^{P_i}_t(\bar\theta_{i-1})< D\;, \\
 \bar\theta_{i-1} &\text{ if } \ell^{P_i}_t(\bar\theta_{i-1})\geq D\;.
 \end{cases} \]
If $\theta$ is a fixed point of $g$, then $f$ defined by $f_{P_i}^+(\theta)=u$ for $\theta \in [\bar\theta_{i-1},\bar\theta_i)$ and 0 otherwise for all $P_i \in \PP$ is a layered Nash flow by definition of $g$.
\end{proof}

\subsection{Parallel path networks}

We say a graph is a parallel path network if all $s,t$-paths are edge disjoint. We say a graph is a parallel link network if every edge is of the form $(s,t)$. In a parallel link network, every path in $\PP$ consists of a single edge, and therefore, we sometimes write $e$ instead of $P$ whenever $P = \{e\}$.

In the next section we will prove tight constant bounds for the inefficiency of Nash flows in parallel path networks. To do this, we show that in layered Nash flows the least amount of flow arrives at the sink in a given time horizon among all Nash flows in these networks. For a given deadline $D$, let $M_f$ denote the amount of flow that arrives under $f$ at the sink before the deadline.

\begin{theorem}\label{lem:par}
Consider a parallel path network. There exists a layered Nash flow over time $f$ with $M_f \leq M_{f'}$ for all Nash flows over time $f'$.
\end{theorem}
\begin{proof}
By \cite[Lemma 7]{correa2019price} we can reduce a parallel path network to a parallel link network without changing the flow dynamics. Therefore, to prove results in parallel path networks, it suffices to prove them in parallel link networks.

Consider a parallel link network. Without loss of generality we assume that $\nu_e\leq u$ for all $e\in E$. Consider a Nash flow $f$. We show the theorem by proving that there exists a layered Nash flow $f''$ with $M_{f''}\leq M_f$. Let $\bar{\theta}$ be the first particle that does not arrive before $D$ under $f$. Observe that $M_f = u \bar{\theta}$. If $\bar{\theta} = D$, we immediately conclude that $f$ is an optimal flow and thus there exists a layered Nash flow $f''$ with $M_{f''}\leq M_f$.

Let us assume $\bar{\theta} < D$. Recall that \cref{cor:Phases} implies that $f$ consists of different phases, where in each phase $\phi$ it sends inflow into a set of edges $E_{\phi}$ having the same costs. We will transform $f$ step by step into in intermediary layered Nash flow $f'$ for non-constant inflow with $M_{f'} \leq M_f$. We initialize $f' = f$ and for each phase in which $E_{\phi}$ is not a singleton, we apply the following iterative procedure to update $f'$.

We call an edge $e_i \in E_{\phi}$ \emph{full} if there is an $\varepsilon>0$ such that $f_{e_i}^-(\theta) = \nu_{e_i}$ for all $\theta \in [D-\varepsilon, D)$. Let $\varepsilon_i$ be the maximum such value of $\varepsilon$ for each full edge $e_i$. Note that in the Nash flow $f$ we have $\tau_{e'} \geq D-\bar{\theta}$ for all edges $e'$ that are not full. Otherwise, the particle leaving the source at time $\bar{\theta}$ could take edge $e'$ and strictly improve its cost.

For each full edge $e_i \in E_{\phi}$, define $a_i=D-\varepsilon_i-\tau_{e_i}$.
Let $e_1$ be the edge with the smallest value of $a_i$ and denote $\tau_1 = \tau_{e_1}$ and $\nu_1 = \nu_{e_1}$ and let $\hat{\theta}_1=\sup\{\theta\mid f_{e_1}^+(\theta)>0\text{ and }\ell_t^{e_1}(\theta)<D\}$.
Define $b_1 = a_1 + \varepsilon_1\cdot \nu_1/u$ and let $I_1 = \left[a_1, b_1\right)$. We will adjust $f'$ for $\theta\in[a_1,b_1)$ and later also for $\theta \in [b_1,\hat{\theta}_1)$. First, set $f_{e_1}^{\prime+}(\theta)=u$ and $f_{e'}^{\prime+}(\theta)=0$ for all $e'\neq e_1$, for all $\theta\in I_1$. $I_1$ is called a \emph{block} for edge $e_1$. In words, in a block for $e_1$ all flow is sent only into $e_1$. Additionally, set $f_{e_1}^{\prime+}(\theta) = 0$ for all $\theta \in [b_1, \hat{\theta}_1)$. Observe that the amount of flow on $e_1$ in the block $I_1$ is chosen such that $f_{e_1}^{\prime-}(\theta) = \nu_1$ for all $\theta \in [D-\varepsilon_i, D)$. Furthermore, note that $[a_1,b_1)$ is the shortest time interval with this property and therefore $b_1 \leq \hat{\theta}_1$.

Now consider the set of edges $E_1 \subseteq E_{\phi} \setminus \{e_1\}$ which had a measurable positive inflow in $I_1$ under $f$. Note that the amount of flow sent into $e_1$ in $[a_1,\hat{\theta}_1)$ has not changed. Therefore, we can send the same amount of flow into each edge in $E_1$ within $[a_1,\hat{\theta}_1)$.  
To achieve this, in $f'$, we divide the flow particles that were sent into $e_1$ under $f$ in $[b_1, \hat{\theta}_1)$ to the edges in $E_1$ proportionally to the amount of flow deleted during $[a_1,b_1)$. 

Comparing $f$ and $f'$ for an edge $e'\in E_1$, we removed some inflow in $[a_1,b_1)$ and added the same amount in $[b_1,\hat{\theta}_1)$. Given that the same amount of flow will enter $e'$ and we only postponed some particles, this transformation yields a queue on $e'$ at time $\hat{\theta}_1$ that is at least as large as under $f$. In case some particles do not arrive at $t$ before $D$ due to this transformation, we will restore the Nash flow over time condition later.

Now we repeat the above procedure for all edges $E_{\phi}\setminus \{e_1\}$ with positive inflow, making further changes to the current flow $f'$. This starts with defining full edges and recomputing the values of $\varepsilon_i$ and $a_i$ for all $e_i \in E_{\phi} \setminus \{e_1\}$.
Observe that because we choose $a_i$ as small as possible in every iteration, the blocks are created in chronological order and therefore we do not change blocks that have previously been created.

After the iterative procedure for phase $\phi$, delete all inflow into the edges in $E_{\phi}$, except for the flow sent in the blocks. Flow is deleted by setting the inflow into the network equal to 0 at the corresponding time interval. The inflow into each edge in $E_{\phi}$ therefore consists of one block or is always $0$. This terminates the procedure for the current phase.

We continue by applying the procedure to the next phase. After completing the procedure for all phases, all particles that are not deleted arrive at the sink in time. Additionally, no particle has an incentive to switch to an edge with a lower cost because all those edges are full or too long and thus it will not arrive at the sink before the deadline. Thus, the original flow $f$ has been transformed into a layered Nash flow $f'$ with non-constant inflow. In particular, the particle that departed at time $\bar{\theta}$ in $f$ will still not arrive before $D$ in $f'$, so $\bar{\theta}_{f'} \leq \bar{\theta}_f$. 

We finish the proof by turning $f'$ into a layered Nash flow $f''$ with the same constant inflow as the original flow, such that $\bar{\theta}_{f''} \leq \bar{\theta}_{f'}$. We provide an algorithmic proof for this step in \cref{clm:TurnIntoLayered}.
\end{proof}

\begin{claimex}\label{clm:TurnIntoLayered}
Applying \cref{alg:Monotonicity} on the current Nash flow $f'$ with non-constant inflow turns it into a layered Nash flow $f''$ such that $\bar{\theta}_{f''} \leq \bar{\theta}_{f'}$.
\end{claimex}

\begin{proof}
We will start the proof of the correctness of the algorithm with a proof by induction over the while-loop. Let $e$ be the edge with the last inflow into it that arrives before the deadline, i.e., $e=\argmin_{e \in S} \{ \sup \{ \theta: f_e^{\prime+}(\theta)>0, \ell_t^e(\theta) < D \} \}$, and define $\underline{\theta}_e=\min\{\theta:f_e^{\prime+}(\theta)>0\}$ and $\bar{\theta}_e=\sup\{\theta:f_e^{\prime+}(\theta)>0, \ell_t^e(\theta) < D\}$. Note that $f_e^{\prime+}(\theta) = u$ for $\theta \in [\underline{\theta}_e,\bar{\theta}_e)$. Consider an iteration of the while-loop and define $\underline{T}$ as the value of $T$ at the start of this iteration and $\bar{T}$ as the value of $T$ at the end of this iteration. We will prove that $\bar{T}\leq \bar{\theta}_e$.

By the induction hypothesis, we can assume that $\underline{T}\leq \underline{\theta}_e$. Recall that for all $e'\in V$, we have that $\tau_{e'}\geq D-\underline{\theta}_e$. If not, $f'$ was not a layered Nash flow as particle $\underline{\theta}_e$ could improve by switching to $e'$. If there is an edge $e'\in V$ with $\ell_t^{e'}(\underline{T})<D$, then $\underline{T}$ gets updated to $T'=\underline{T}+(D-\underline{T}-\tau_{e'}) \nu_{e'} / u = \underline{T}(1-\nu_{e'}/u)+(D-\tau_{e'}) \nu_{e'} / u \leq \underline{\theta}_e(1-\nu_{e'}/u)+\underline{\theta}_e \nu_{e'}/u=\underline{\theta}_e$. Similarly, $T' \leq \underline{\theta}_e$ for all later edges $e'\in V$ with $\ell_t^{e'}(\underline{T})<D$.

This implies that at the start of \texttt{fillEdge}$(e,T)$, $T \leq \underline{\theta}_e$ and thus we have that $\bar{T}$ at the end of the while-loop is at most $T+(D-T-\tau_e) \nu_e / u\leq \underline{\theta}_e+(D-\underline{\theta}_e-\tau_e) \nu_e / u=\bar{\theta}_e$, as the inflow is positive only in the block on the full edge. The proof by induction on the while-loop is complete.

\begin{algorithm}[b!]
  \SetAlgoLined\DontPrintSemicolon
  \SetKwFunction{algo}{createLayeredFlow}\SetKwFunction{proc}{fillEdge}
  \SetKwProg{myalg}{Algorithm}{}{}
  \myalg{\algo{f'}}{
 \nl $T \gets 0$\;
 \nl $S \gets \{e \mid \exists \; \theta: f_e^{\prime +}(\theta) > 0 \}$\;
 \nl $U \gets \{e \in E \setminus S\}$\;
 \nl $V \gets \emptyset$\;
 \nl \While{$S \neq \emptyset$}{
  \nl $e \gets \argmin_{e \in S} \{ \sup \{ \theta: f_e^{\prime +}(\theta)>0\} \}$\;
  \nl $V \gets \{e' \in U \mid c_{e'} < c_e \}$\;
  \nl \For{$e' \in V$ in non-descending order of edge costs}{
    \nl \proc{e',T}\;
    \nl $U \gets U \setminus \{e'\}$\;
  }
  \nl \proc{e,T}\;
  \nl $S \gets S \setminus \{e\}$\;
 }
 \nl \For{$e' \in U$ in non-descending order of edge costs}{
    \nl \proc{e',T}\;
    \nl $U \gets U \setminus \{e'\}$\;
  }
  \nl $f_e^{\prime \prime +}(\theta) \gets u \text{ for }\theta \in [T,\infty) \text{ for some arbitrarily chosen } e\in E $\;
  \nl \Return{$f''$}\;
 }
 \BlankLine
 \BlankLine
 \setcounter{AlgoLine}{0}
  \SetKwProg{myproc}{Subroutine}{}{}
  \myproc{\proc{e,T}}{
 \nl \eIf{ $\ell_t^e(T) < D$}{
 \nl $T' \gets T + (D-T-\tau_e) \nu_e / u$\;
 \nl $f_e^{\prime \prime +}(\theta) \gets u \text{ for } \theta \in [T,T')$\;
 \nl $f_e^{\prime \prime +}(\theta) \gets 0 \text{ for } \theta \in [0,T) \cup [T',\infty)$\;
 \nl $T \gets T'$\;
 }{\nl $f_e^{\prime \prime +}(\theta) \gets 0 \text{ for } \theta \in [0,\infty)$\;}
 }
 \caption{Algorithm for \cref{clm:TurnIntoLayered}.}\label{alg:Monotonicity}
\end{algorithm}

After the while loop there might be some edges left in $U$ that did not receive flow in $f$. For an edge $e \in U$ we have that if $\ell_t^e(T)<D$, the value of $T$ remains at most $\bar{\theta}_{f'}$, similar to the previously considered edges in $U$.

To prove that the resulting flow $f''$ is a Nash flow, consider a time $\theta$ and an edge $e$ such that $f_e^{\prime\prime +}(\theta) > 0$. Note that by the iterative procedure that we applied on $f'$, the sequence of edges that are selected in line 6 of the algorithm are in non-decreasing order of edge costs. If there are edges that did not receive flow in $f'$ with a lower cost than $c_e$, these edges are considered in the algorithm before edge $e$. Therefore by the \texttt{fillEdge} subroutine all edges with lower costs than $c_e$ are either full or too long. Hence, edge $e$ was indeed the cheapest edge that the particle at time $\theta$ could take such that it arrives at the sink before the deadline. The fact that $f''$ is a layered Nash flow follows from the fact that the \texttt{fillEdge} subroutine sends inflow into each edge such that the last particle arrives exactly before the deadline.

We conclude that $f''$ is a layered Nash flow such that $\bar{\theta}_{f''} \leq \bar{\theta}_{f'}$. Therefore, the proof of \cref{clm:TurnIntoLayered} is complete.
\end{proof}


\section{Price of Anarchy}\label{sec:PoA}

We use two different measures to quantify the quality of a flow over time $f$. First, we consider the throughput objective under which we want to maximize the amount of flow that can reach the sink before the given deadline $D$. Recall that for a flow $f$, $M_f$ denotes the amount of flow that arrives at the sink before the deadline. Let $M^*= \max_f M_f$ be the maximum amount of flow that arrives at the sink before the deadline over all possible flows over time. We denote by $\mathcal{N}$ the set of Nash flows of a given instance. The \emph{throughput-price of anarchy} of a given instance is defined by
\[t\text{-}PoA=\max_{f\in \mathcal{N}}\frac{M^*}{M_f} \; .\]

Second, we consider the makespan objective. For the given deadline $D$, let $M= \min_{f \in \mathcal{N}} M_f$ be the smallest amount of flow that arrives in time over all Nash flows. We are interested in how fast we could possibly send $M$ to the sink. We define $D^*(M)$ as the earliest point in time for which a flow over time $f$ exists such that $M$ units of flow arrive at the sink before $D^*(M)$. The \emph{makespan-price of anarchy} of a given instance is defined by
\[m\text{-}PoA=\frac{D}{D^*(M)} \; .\]

\subsection{General networks}

For general networks, both prices of anarchy turn out to be unbounded.

\begin{theorem}\label{thm:unb}
Even in series-parallel networks, the throughput-price of anarchy and the makespan-price of anarchy are unbounded.
\end{theorem}
\begin{proof}
First, we consider the throughput-PoA. We construct a sequence of instances depicted in \cref{fig:PoAInfinite} that for some $k\in \mathbb{N}$ consists of $k$ subnetworks connected in series. Each subnetwork consists of a two-edge parallel link network connected in series with a single edge. The transit time of every upper edge in the parallel composition is $1$, the transit time of every lower edge in the parallel composition is $0$, and the transit time of every single edge is $0$. The inflow rate is $u=1$ and the capacity of every edge is equal to $1$. The cost of every upper and single edge is equal to $0$, the cost of the $i$-th lower edge is equal to $2^{k-i}$. We name paths such that $\sum_{e\in P_1}c_e<\sum_{e\in P_2}c_e<\ldots<\sum_{e\in P_{2^k}}c_e$. In particular, $P_1$ consists of the path containing all upper edges and $P_{2^k}$ contains all lower edges. We choose $D=k+2$.

Consider the layered flow over time $f$ defined by $f_{P_i}^+(\theta) = 1$ for $\theta \in \left[\frac{i-1}{2^{k-1}},\frac{i}{2^{k-1}}\right)$ and $0$ otherwise for all $i=1,\ldots 2^k$. For all $\theta\in[2,\infty)$, we allocate particles arbitrarily, as they will not arrive before $D$. Notice that the inflow into the $i$-th subnetwork is positive in the interval $[i-1,i+1)$, and the particles in $[i-1,i)$ take the upper edge and the particles in $[i,i+1)$ take the lower edge implying a queuing time of $1$ on the single edge. This implies that the last particle of every path arrives at the sink just before time point $D$. One can prove that $f$ is a Nash flow with $M_f = 2$.

The maximum amount of flow that can be sent to the sink with a deadline $D=k+2$ is equal to $k+2$ along $P_{2^k}$. This yields a throughput-price of anarchy of at least $1+k/2$. Regarding the makespan-price of anarchy, observe that an amount of flow of $M_f=2$ units can be sent over the network in $2$ time units. The constructed Nash flow yields a makespan-price of anarchy of at least $1+k/2$, which shows they are both unbounded.
\end{proof}

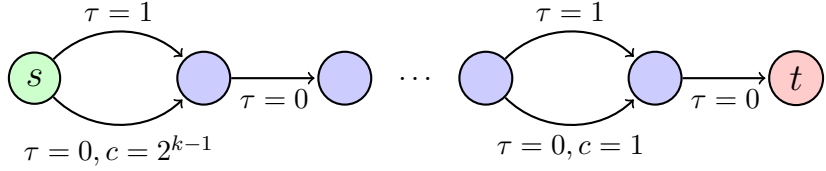
\begin{figure}[t]
\centering
\begin{tikzpicture}[->,shorten >=1pt,auto,node distance=2cm,
  thick,main node/.style={circle,fill=blue!20,draw,minimum size=20pt,font=\sffamily\Large\bfseries},source node/.style={circle,fill=green!20,draw,minimum size=15pt,font=\sffamily\Large\bfseries},dest node/.style={circle,fill=red!20,draw,minimum size=15pt,font=\sffamily\Large\bfseries},scale=0.75]
\node[source node] (1) at (0,0) {$s$};
\node[main node] (2) at (3,0) {};
\node[main node] (3) at (5.5,0) {};
\node[main node] (4) at (8,0) {};
\node[main node] (5) at (11,0) {};
\node[dest node] (6) at (13.5,0) {$t$};
\node  at (6.75,0) {$\ldots$};
\draw (1) to [bend left=45] node[above] {$\tau=1$} (2);
\draw (1) to [bend right=45] node[below] {$\tau=0, c=2^{k-1}$} (2);
\draw (2) to node[below] {$\tau=0$} (3);
\draw (4) to [bend left=45] node[above] {$\tau=1$} (5);
\draw (4) to [bend right=45] node[below] {$\tau=0, c=1$} (5);
\draw (5) to node[below] {$\tau=0$} (6);
\end{tikzpicture}
\caption{An example showing unbounded prices of anarchy in series-parallel network.}
\label{fig:PoAInfinite}
\end{figure}

\subsection{Parallel path networks}

By the previous example, a natural candidate for a graph class with a finite price of anarchy bound is parallel path networks. In this subsection we restrict our attention to these graphs. Recall that by \cite{correa2019price} it is sufficient to prove the results for parallel link networks.

\subsubsection{Throughput-price of anarchy.}

First, we provide an asymptotically tight bound of $\e/(\e-1)$ for parallel graphs with no transit times. Second, we give a tight bound of $2$ for parallel graphs with arbitrary transit times.

\begin{theorem}\label{thm:0}
In a parallel path network with $\tau_e=0$ for all $e\in E$ the throughput price of anarchy is at most $\e/(\e-1)$.
This bound is asymptotically tight.
\end{theorem}
\begin{proof}

By Lemma \ref{lem:par}, we can restrict attention to layered Nash flows over time. So assume that $f$ is a layered Nash flow over time in a parallel link network with edges $e_1, \ldots, e_k$. Let $0=\theta_0\leq \theta_1\leq \ldots \leq\theta_{k}\leq D$ be as specified in \cref{def_layeredNash}.
Without loss of generality assume that $u=1$ and $\nu_e\leq 1$ for all $e\in E$.
Given that the last particle on an edge 
arrives exactly at the deadline, we have that for all $i=1,\ldots,k,$
\begin{equation*}\label{eq:recu}
\theta_{i-1}+\frac{\theta_i-\theta_{i-1}}{\nu_i}=D\;.
\end{equation*}
Rewriting yields 
\begin{equation} \label{eq:rec}
\theta_i=(1-\nu_i)\cdot \theta_{i-1}+\nu_i \cdot D\;.
\end{equation}

\begin{claimex}\label{cla:the}
Solving \cref{eq:rec} with $\theta_0 = 0$ yields for all $i=1,\ldots,k,$
\[\theta_i=D\cdot\left(1-\prod_{j=1}^i(1-\nu_j)\right)\;.\]
\end{claimex}

\begin{proof}
We use a proof by induction. Observe that for $i=1$, we have that $\theta_1=D\cdot \nu_1$. Assume the statement is true for edge $i-1$, then for edge $i$ we have that
\begin{align*}
\theta_i&=(1-\nu_i)\cdot \theta_{i-1}+\nu_i \cdot D\\
&= (1-\nu_i)\cdot D\cdot\left(1-\prod_{j=1}^{i-1}(1-\nu_j)\right)+\nu_i \cdot D\\
&=D\cdot\left(1-\prod_{j=1}^{i}(1-\nu_j)\right)\;,
\end{align*}
where the second equality follows from the induction hypothesis.
\end{proof}

Observe that the socially optimal flow sends $\min\{1,\sum_{i=1}^k\nu_i\}\cdot D$ units of flow before the deadline and the Nash flow over time sends $\theta_k$ units of flow before the deadline. Therefore, $t\text{-}POA = \min\{1,\sum_{i=1}^k\nu_i\}\cdot D / \theta_k$. In order to maximize its value we can assume that $\nu_i<1$ for all $i=1,\ldots,k$, as otherwise $t\text{-}PoA=1$. 

We now argue that we obtain the worst case price of anarchy for $\sum_{i=1}^k\nu_i=1$.

\begin{claimex}\label{claim:MaxPoA}
$t\text{-}PoA = \min\{1,\sum_{i=1}^k\nu_i\}\cdot D / \theta_k$ is maximized if $\sum_{i=1}^k\nu_i=1$.
\end{claimex}

\begin{proof}
If $\sum_{i=1}^k\nu_i>1$, observe that by decreasing the capacity of some edge $i$, the amount of flow that reaches $t$ in the Nash flow over time decreases as $\frac{\partial\theta_k}{\partial\nu_i}=D\cdot\prod_{j\neq i}(1-\nu_j)>0$ for all $i=1,\ldots,k$ and the optimal flow stays constant. This implies we increase the throughput-$PoA$. 

In the other case, if $\sum_{i=1}^k\nu_i<1$, observe that for each edge $i$
\begin{align}
\frac{\partial t\text{-}PoA}{\partial \nu_i}&=\frac{\partial}{\partial \nu_i} \frac{\sum_{j=1}^k\nu_j}{1-\prod_{j=1}^k(1-\nu_j)}\notag\\ 
&=\frac{1-\left(\prod_{j=1}^k(1-\nu_j)+\sum_{j=1}^k\nu_j\cdot \prod_{j\neq i}(1-\nu_j)\right)}{\left(1-\prod_{j=1}^k(1-\nu_j)\right)^2} \label{eq:partialderivativePoA} \\ 
&=\frac{1-\left(1+\sum_{j\neq i}\nu_j\right)\cdot \prod_{j\neq i}(1-\nu_j)}{\left(1-\prod_{j=1}^k(1-\nu_j)\right)^2}\;. \notag
\end{align}
In order to prove that $\frac{\partial t\text{-}PoA}{\partial \nu_i}>0$, it remains to show that $\left(1+\sum_{j\neq i}\nu_j\right)\cdot \prod_{j\neq i}(1-\nu_j)<1$. Since for all $\ell\neq i$, we have that
\begin{align*}
\frac{\partial}{\partial \nu_{\ell}} \left(1+\sum_{j\neq i}\nu_j\right)\cdot \prod_{j\neq i}(1-\nu_j)&=\prod_{j\neq i}(1-\nu_j)-\left(1+\sum_{j\neq i}\nu_j\right)\cdot \prod_{j\neq i,\ell}(1-\nu_j)\\
&=-\left(2\nu_{\ell}+\sum_{j\neq i,\ell}\nu_j\right)\cdot \prod_{j\neq i,\ell}(1-\nu_j)<0\;,
\end{align*}
we know that the maximum of $\left(1+\sum_{j\neq i}\nu_j\right)\cdot \prod_{j\neq i}(1-\nu_j)$ is attained if and only if $\nu_{\ell}=0$ for all $\ell\neq i$. However, since $\nu_e>0$ for all $e\in E$, this implies that $\left(1+\sum_{j\neq i}\nu_j\right)\cdot \prod_{j\neq i}(1-\nu_j)<1$ and hence $\frac{\partial t\text{-}PoA}{\partial \nu_i}>0$. This implies we increase the throughput-$PoA$ by increasing the capacity of edge $i$.
\end{proof}

Using the expression for $\theta_k$ from \cref{cla:the} and the result of \cref{claim:MaxPoA}, we want to solve the following optimization problem in order to maximize the throughput-price of anarchy:
\begin{align*}
\max\: &\prod_{j=1}^k(1-\nu_j)\\
\text{s.t.}&\sum_{i=1}^k\nu_i=1\;,\\
&\nu_i\geq 0\text{ for all }i=1,\ldots,k\;.
\end{align*}
This problem is solved by setting $1-\nu_1=\ldots=1-\nu_k$ and hence 
$\nu_1=\ldots=\nu_k=1/k$. This implies that
\[t\text{-}PoA\leq\frac{1}{1-\left(\frac{k-1}{k}\right)^k}\leq  \frac{\e}{\e-1}\;,\]
where the last inequality follows as $\frac{1}{1-\left(\frac{k-1}{k}\right)^k}$ increases with $k$. The fact that this bound is tight follows from \cref{ex:tight}, which completes the proof of \cref{thm:0}.
\end{proof}

\begin{example}\label{ex:tight}
\cref{fig:ee-1istight} with $\tau_e=0$ for all $e\in E$, $u = D=1$, and $k \rightarrow \infty$ gives an asymptotically tight example for \cref{thm:0}.\qedhere
\end{example}

\begin{figure}[h]
\centering
\begin{tikzpicture}[->,shorten >=1pt,auto,node distance=2cm,thick,main node/.style={circle,fill=blue!20,draw,minimum size=20pt,font=\sffamily\Large\bfseries},source node/.style={circle,fill=green!20,draw,minimum size=15pt,font=\sffamily\Large\bfseries},dest node/.style={circle,fill=red!20,draw,minimum size=15pt,font=\sffamily\Large\bfseries},scale=0.75]
\node[source node] (1) at (0,0) {$s$};
\node[dest node] (2) at (4,0) {$t$};
\draw (1) to [bend left=45] node[above] {$\nu_1=1/k$} (2);
\draw (1) to node[above] {$\nu_2=1/k$} (2);
\node at (2,-0.3) {$\vdots$};
\draw (1) to [bend right=45] node[below] {$\nu_k=1/k$} (2);
\end{tikzpicture}
\caption{The tight example for \cref{thm:0,PoA2Parallel} with a throughput-$PoA$ and makespan-$PoA$ arbitrarily close to $\e/(\e-1)$.}
\label{fig:ee-1istight}
\end{figure}
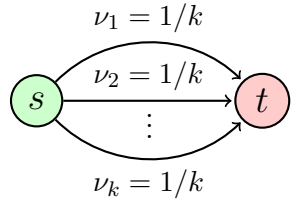

Our next result determines the throughput-price of anarchy in parallel path networks with arbitrary transit times.

\begin{theorem}\label{thm:PoAis2}
In a parallel path network the throughput-price of anarchy is at most $2$. 
This bound is tight.
\end{theorem}

\begin{proof}
As before, assume $f$ is a layered Nash flow over time in a parallel link network with edges $e_1, \ldots, e_k$.
Without loss of generality assume that $u=1$ and $\nu_e\leq 1$ for all $e\in E$.
We will derive some properties of the network so that the price of anarchy is maximized.

Observe that if $\theta_{i-1}=\theta_i$ for some edge $i$, we have that $\tau_i\geq D-\theta_{i-1}$. Because no positive flow mass is sent along edge $i$ if $\theta_{i-1}=\theta_i$ and we want to maximize the throughput-$PoA$, we can assume that $\tau_i=D-\theta_{i-1}$, as a larger transit time in this case will only decrease the optimal flow, but not affect the Nash flow over time. Thus $\tau_i \leq D - \theta_{i-1}$ for all $i$, which implies that
\begin{equation*}
    \theta_{i-1}+\frac{\theta_i-\theta_{i-1}}{\nu_i}+\tau_i=D\;. 
\end{equation*}
Rewriting yields
\begin{equation}
\theta_i=(1-\nu_i)\cdot \theta_{i-1}+\nu_i \cdot (D-\tau_i) \; . \label{eq:recursive2}
\end{equation}
\begin{claimex}\label{cla:tht}
Solving \cref{eq:recursive2} with $\theta_0 = 0$ yields for all $i=1,\ldots,k,$
\begin{equation*}
\theta_i
=D\cdot \left(1-\prod_{j=1}^{i}(1-\nu_j)\right) - \sum_{j=1}^i \left( \tau_j\cdot\nu_j\cdot\prod_{\ell=j+1}^i(1-\nu_{\ell}) \right)\;.
\end{equation*}
\end{claimex}
\begin{proof}
We use a proof by induction. Observe that for $i=1$, we have that $\theta_1=(D-\tau_1)\cdot\nu_1$. Assume the statement is true for edge $i-1$, then for edge $i$ we have that
\begin{align*}
\theta_i&=(1-\nu_i)\cdot \theta_{i-1}+\nu_i \cdot (D-\tau_i)\\
&= (1-\nu_i)\cdot \left(D\cdot\left(1-\prod_{j=1}^{i-1}(1-\nu_j)\right)-\sum_{j=1}^{i-1} \left( \tau_j\cdot\nu_j\cdot\prod_{\ell=j+1}^{i-1}(1-\nu_{\ell})\right) \right)+\nu_i \cdot (D-\tau_i)\\
&=D\cdot\left(1-\prod_{j=1}^{i}(1-\nu_j)\right)-\sum_{j=1}^i \left( \tau_j\cdot\nu_j\cdot\prod_{\ell=j+1}^i(1-\nu_{\ell}) \right)\;,
\end{align*}
where the second equality follows from the induction hypothesis.
\end{proof}

Similar to \cref{claim:MaxPoA} we can show that we can assume that $\sum_{i=1}^k\nu_i\leq 1$.

\begin{claimex}\label{claim:MaxPoA2}
$t\text{-}PoA$ is maximized if $\sum_{i=1}^k\nu_i\leq 1$.
\end{claimex}
\begin{proof}
Let $i$ denote the edge with largest transit time. Suppose the claim is not true, then by slightly reducing the capacity of the edge $i$, we do not change the optimal flow.
Now we investigate the effect of reducing the capacity of this edge on the values of the $\theta_i$ sequence in the Nash flow over time. 
Note that
$\frac{\partial \theta_i}{\partial \nu_i} = - \theta_{i+1} + (D-\tau_i) \geq 0$. Therefore, reducing $\nu_i$ changes the value of $\theta_i$ to $\theta_i' \leq \theta_i$. Note that due to the structure of the layered Nash flow over time, $\theta_j' = \theta_j$ for all $j < i$.

We now prove by induction that $\theta_j' \leq \theta_j$ for all $j > i$. Assume the statement is true up to some $\theta_j$ for a fixed $j$ (initially, it is true for $j=i$), so $\theta_j' \leq \theta_j$. If edge $j+1$ carries flow, then by the structure of the layered Nash flow over time we have that
\begin{equation}
    \theta_j' + \frac{\theta_{j+1}' - \theta_j'}{\nu_{j+1}} + \tau_{j+1} = D \; . \label{eq:recursive3}
\end{equation}
Rearranging terms yields $\theta_{j+1}' = (1 - \nu_{j+1}) \theta_j' + \nu_{j+1} (D - \tau_{j+1})$. Combining this with \cref{eq:recursive2} and $\theta_j' \leq \theta_j$ yields that $\theta_{j+1}' \leq \theta_{j+1}$.

On the other hand, if edge $j+1$ does not carry flow, we have $\theta'_{j+1} = \theta'_j$.
Thus, $\theta_{j+1}' = \theta'_j \leq \theta_j \leq \theta_{j+1}$.

Therefore, reducing the capacity of the edge with the longest transit time improves the quality of the Nash flow over time but not of the optimal flow, and hence the throughput-price of anarchy is maximized if $\sum_{i=1}^k \nu_i \leq 1$.
\end{proof}
From the above claims we conclude that the throughput-$PoA$ is equal to
\begin{align*}\frac{M^*}{M_f} &= \frac{\sum_{j=1}^k \nu_j (D-\tau_j)}{\theta_{k}}\\
&= \frac{\sum_{j=1}^k \nu_j (D-\tau_j)} {D\cdot\left(1-\prod_{j=1}^{k}(1-\nu_j)\right)-\sum_{j=1}^{k}\tau_j\cdot\nu_j\cdot\prod_{\ell=j+1}^{k}(1-\nu_{\ell})}\;.
\end{align*}
Taking the derivative with respect to $\tau_i$ for some $i \leq k$ we get
\begin{align*}
\frac{\partial t\text{-}PoA}{\partial \tau_i} &= \frac{-\nu_i\left(D\cdot\left(1-\prod_{j=1}^{k}(1-\nu_j)\right)-\sum_{j=1}^{k}\tau_j\cdot\nu_j\cdot\prod_{\ell=j+1}^{k}(1-\nu_{\ell})\right)}{\left(D\cdot\left(1-\prod_{j=1}^{k}(1-\nu_j)\right)-\sum_{j=1}^{k}\tau_j\cdot\nu_j\cdot\prod_{\ell=j+1}^{k}(1-\nu_{\ell})\right)^2}\\
&- \frac{\sum_{j=1}^k \nu_j (D-\tau_j) \left(-\nu_i \prod_{\ell=i+1}^{k}(1-\nu_\ell)\right)}{\left(D\cdot\left(1-\prod_{j=1}^{k}(1-\nu_j)\right)-\sum_{j=1}^{k}\tau_j\cdot\nu_j\cdot\prod_{\ell=j+1}^{k}(1-\nu_{\ell})\right)^2}\;.
\end{align*}
We observe that the numerator of the derivative is independent of $\tau_i$, so to achieve a maximum throughput-$PoA$, we either set $\tau_i=0$ or $\tau_i =D- \theta_{i-1}$ for all $i$. This implies that for an edge $i$ either $\tau_i = 0$ or the edge does not carry flow in the Nash flow over time. In particular, we can prove the following structure.

\begin{claimex}\label{clm:NoTau}
We can assume that $\tau_i = 0$ for all $i=1, \ldots, k-1$ and $\tau_k = D - \theta_{k-1}$.
\end{claimex}
\begin{proof}
We first argue that $\tau_1=0$. Suppose this is not the case. As $\theta_0=0$ the above discussion would imply that $\tau_1=D$, which means that both the optimal flow and the Nash flow cannot send a positive flow mass along edge $e_1$ and we can delete the edge without changing the throughput-$PoA$. So we can assume $\tau_1 = 0$.

Now we show that $\tau_2 = \tau_3 = \dots = \tau_{k-1}=0$, and $\tau_{k} = D - \theta_{k-1}$. First suppose that there are edges $i$, $i+1$ with $\tau_i\neq 0$ and $\tau_{i+1} = 0$. In this case, we change the order of edges $i$ and $i+1$ in the layered Nash flow over time and consider the order $(1,\dots,i-1, i+1, i, i+2, \dots, k)$. Since $\tau_i \neq 0$, edge $i$ does not carry any flow, and therefore $\theta_{i-1} = \theta_i$, which means that switching edges $i$ and $i+1$ does not change the amount of flow sent along edge $i+1$. Additionally, as $\tau_i = D-\theta_{i-1}$, and edge $i$ only starts to receive flow from time $\theta_{i+1}$ onwards in the new Nash flow, edge $i$ still does not receive any flow. Hence $\theta_{i+1}, \dots, \theta_{k}$ remain the same. We conclude that we can assume $\tau_1=\dots=\tau_{k'}=0$ and $\tau_{k'+1}=\dots=\tau_k=D-\theta_{k'}$ for some $k'$.

It remains to argue that we can assume that $k'+1=k$. Since these edges do not receive any flow in the Nash flow and have the same transit time, we can merge these edges into one edge. This changes neither the amount of flow sent in the Nash flow over time nor in the optimal solution.
\end{proof}

We continue by explicitly stating the throughput-$PoA$.
\begin{align*}
    t\text{-}PoA &= \frac{\sum_{j=1}^{k}(D-\tau_j)\cdot\nu_j }{\theta_k}=\frac{D\cdot\sum_{j=1}^{k-1}\nu_j+\theta_{k-1}\cdot\nu_k }{\theta_{k-1}}\\
    &=\frac{\sum_{j=1}^{k-1}\nu_j }{1-\prod_{j=1}^{k-1}(1-\nu_j)}+\nu_k \; ,
\end{align*}
where the second inequality follows from $\tau_1=\ldots=\tau_{k-1}=0$ and $\tau_k=D-\theta_{k-1}$ and the fact that $\theta_k=\theta_{k-1}$, and the third equality follows from the fact that $\theta_{k-1}$ is as defined in \cref{cla:the}.

With \cref{clm:NoTau} at our disposal, we can now show that we can strengthen \cref{claim:MaxPoA2} to the following.

\begin{claimex}\label{claim:MaxPoA2Equal}
$t\text{-}PoA$ is maximized if $\sum_{i=1}^k\nu_i = 1$.
\end{claimex}
\begin{proof}
    This follows from \cref{eq:partialderivativePoA} in the proof of \cref{thm:0} since the partial derivative with respect to $\nu_i$ is larger than $0$ for each edge $i$. As before, we will show the edge capacities are all equal. In the following, we assume that $\nu_k$ is already fixed and show that $\nu_1=\nu_2 = \dots = \nu_{k-1}$. As $\sum_{j=1}^{k-1}\nu_j = 1- \nu_k$, the numerator of the throughput-$PoA$ term is fixed. Thus, the throughput-$PoA$ is maximized for maximum $\prod_{j=1}^{k-1}(1-\nu_j)$. Following similar arguments as in the proof of \cref{thm:0}, we can conclude that $(1-\nu_i) = (1-\nu_j)$ for all $i,j$, i.e., that $\nu_1 = \nu_2 = \dots = \nu_{k-1}$.
\end{proof}

This implies that the throughput-price of anarchy can be written as
\[t\text{-}PoA=\frac{(k-1)\cdot\nu_1}{1-(1-\nu_1)^{k-1}}+1-(k-1)\cdot\nu_1 \; .\]
In order to maximize this expression, \cref{partialDervPoAnegative} shows that the partial derivative with respect to $\nu_1$ is negative. 

\begin{claimex}
$\frac{\partial}{\partial \nu_1}\left( \frac{(k-1)\cdot\nu_1}{1-(1-\nu_1)^{k-1}}+1-(k-1)\cdot\nu_1\right) < 0 \; .$
\label{partialDervPoAnegative}
\end{claimex}
\begin{proof}
Observe that
\begin{align*}
    \frac{\partial}{\partial \nu_1}& \left(\frac{(k-1)\cdot\nu_1}{1-(1-\nu_1)^{k-1}}+1-(k-1)\cdot\nu_1 \right)\\
    &= \frac{(k-1)(1-(1-\nu_1)^{k-1}) - (k-1)\nu_1 (k-1) (1-\nu_1)^{k-2}}{\left(1-(1-\nu_1)^{k-1}\right)^2} - (k-1)\\
    &= (k-1) \left(\frac{1-(1-\nu_1)^{k-1} - \nu_1 (k-1) (1-\nu_1)^{k-2}}{\left(1-(1-\nu_1)^{k-1}\right)^2} - 1\right)\;,
\end{align*}
i.e., it suffices to show that $1-(1-\nu_1)^{k-1} - \nu_1 (k-1) (1-\nu_1)^{k-2} < \left(1-(1-\nu_1)^{k-1}\right)^2$. Straightforward rewriting shows that this is equivalent to showing that
$1 - k\cdot\nu_1 < (1-\nu_1)^{k}$.
This inequality is also known as Bernoulli's inequality and is true for $k\geq 2$.
\end{proof}

It follows that the throughput-$PoA$ is maximized for $\nu_1 \rightarrow 0$. Using L'H\^{o}pital's rule gives $\lim_{\nu_1 \rightarrow 0} t\text{-}PoA = 1 + 1 = 2$. The fact that this bound is tight follows from \cref{example:PoA2Tight}, which completes the proof of \cref{thm:PoAis2}.
\end{proof}

\begin{example}\label{example:PoA2Tight}
The following example shows that the bound of \cref{thm:PoAis2} is tight. See \cref{fig:PoA2Tight}. Let $u=D=1$ and let the upper edge be cheaper than the lower edge.
The layered flow over time that first sends $\varepsilon$ units over the upper path is a Nash flow; all later particles will arrive after the deadline. The socially optimal flow sends $\varepsilon+\varepsilon\cdot(1-\varepsilon)$ units of flow, yielding a throughput-price of anarchy of
\[t\text{-}PoA\geq\frac{\varepsilon+\varepsilon\cdot(1-\varepsilon)}{\varepsilon}=2-\varepsilon\rightarrow 2 \text{ as }\varepsilon\rightarrow 0\;.\]
\end{example}

\begin{figure}[b]
\centering
\begin{tikzpicture}[->,shorten >=1pt,auto,node distance=2cm,
  thick,main node/.style={circle,fill=blue!20,draw,minimum size=20pt,font=\sffamily\Large\bfseries},source node/.style={circle,fill=green!20,draw,minimum size=15pt,font=\sffamily\Large\bfseries},dest node/.style={circle,fill=red!20,draw,minimum size=15pt,font=\sffamily\Large\bfseries},scale=0.75]
\node[source node] (1) at (0,0) {$s$};
\node[dest node] (2) at (4,0) {$t$};
\draw (1) to [bend left=45] node[above] {$\tau=0, \nu=\varepsilon$} (2);
\draw (1) to [bend right=45] node[below] {$\tau=1-\varepsilon, \nu=1-\varepsilon$} (2);
\end{tikzpicture}
\caption{The tight example for \cref{thm:PoAis2} with a throughput-$PoA$ arbitrarily close to 2 for general transit times.}
\label{fig:PoA2Tight}
\end{figure}
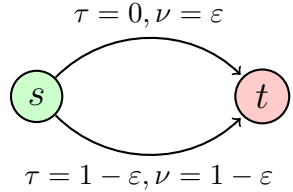

\subsubsection{Makespan-price of anarchy.}

Perhaps surprisingly, the makespan-$PoA$ behaves slightly different: its value is bounded by $\e/(\e-1)$ independent of transit time values.

\begin{theorem}\label{PoA2Parallel}
In a parallel path network the makespan-price of anarchy is at most $\e/(\e-1)$. This bound is tight.
\end{theorem}

\begin{proof}
First recall that, as before, we can restrict our attention to layered Nash flows in parallel link networks with edges $e_1, \ldots, e_k$. Without loss of generality $u=1$ and $\nu_e \leq 1$ for all $e \in E$. Also recall that for the makespan-$PoA$ the goal is to send a fixed amount of flow in the shortest amount of time. Consider a Nash flow $f$ that sends $\theta_k$, as in \cref{cla:tht}, flow to $t$ before a deadline $D$. 
Let $D^*(\theta_k)$ be the optimal deadline for the given amount of flow $\theta_k$, i.e., $D^*(\theta_k)$ is such that $\sum_{j=1}^k \nu_j (D^*(\theta_k)-\tau_j) = \theta_k$. Solving for $D^*(\theta_k)$ gives
\begin{equation*}
    D^*(\theta_k) = \frac{\theta_k + \sum_{j=1}^k \nu_j \tau_j}{\sum_{j=1}^k \nu_j} \; .
\end{equation*}
%
Therefore, the makespan-price of anarchy we want to upper bound is
\begin{align*}
    \frac{D}{D^*(\theta_k)}
    & = \frac{D \cdot \sum_{j=1}^k \nu_j}{\theta_k + \sum_{j=1}^k \nu_j \tau_j} \\
    & = \frac{D \cdot \sum_{j=1}^k \nu_j}{D \left(1 - \prod_{j=1}^k \left(1-\nu_j \right) \right) - \sum_{j=1}^k \left( \nu_j \tau_j \prod_{\ell = j+1}^k \left(1 - \nu_j \right) \right) + \sum_{j=1}^k \nu_j \tau_j} \\
    & =: \frac{\alpha \cdot D}{\beta \cdot D + \gamma} \; , \\
    \textrm{where} & \\
%
%
    \alpha & = \sum_{j=1}^k \nu_j \text{, }
    \beta = 1 - \prod_{j=1}^k \left(1-\nu_j \right) \textrm{, and } 
    \gamma = \sum_{j=1}^{k} \nu_j \tau_j \left(1 - \prod_{\ell = j+1}^k (1-\nu_\ell)\right) \; .
\end{align*}
Observe that $\alpha, \beta, \gamma > 0$, and therefore this is non-decreasing 
in $D$. Hence, the makespan-price of anarchy is maximized for $D \rightarrow \infty$, in which case $\gamma$ vanishes from the ratio. Note that this is why the bound for this price of anarchy notion is independent of the $\tau$ values.

Consequently, the makespan-price of anarchy is maximized by maximizing $\alpha/\beta$. According to \cref{eq:partialderivativePoA} and the following analysis, this can be established by putting $\nu_e = \frac{1}{k}$ for all $e \in E$, which yields a makespan-price of anarchy that has a limit of $\e/(\e-1)$ for $k \rightarrow \infty$. Similarly, \cref{fig:ee-1istight} provides a tight example.
\end{proof}

\paragraph{Acknowledgments.} The authors thank Laura Vargas Koch for early stage discussions.

\bibliographystyle{plain}
\bibliography{references}

\begin{thebibliography}{10}

\bibitem{acemoglu2007competition}
Daron Acemoglu and Asuman Ozdaglar.
\newblock Competition and efficiency in congested markets.
\newblock {\em Mathematics of Operations Research}, 32(1):1--31, 2007.

\bibitem{cominetti2015dynamic}
Roberto Cominetti, Jos{\'e} Correa, and Omar Larr{\'e}.
\newblock Dynamic equilibria in fluid queueing networks.
\newblock {\em Operations Research}, 63(1):21--34, 2015.

\bibitem{cominetti2021longterm}
Roberto Cominetti, Jos{\'e} Correa, and Neil Olver.
\newblock Long-term behavior of dynamic equilibria in fluid queuing networks.
\newblock {\em Operations Research}, 2021.

\bibitem{correa2019price}
Jos{\'e} Correa, Andr{\'e}s Cristi, and Tim Oosterwijk.
\newblock On the price of anarchy for flows over time.
\newblock {\em Mathematics of Operations Research}, forthcoming.

\bibitem{europeancommission}
{European Commission}.
\newblock {\em Urban mobility}, 2021.
\newblock \url{https://ec.europa.eu/transport/themes/urban/urban_mobility_en}
  (accessed October 14, 2021).

\bibitem{eurostat}
Eurostat.
\newblock {\em Passenger cars in the EU}, 2021.
\newblock
  \url{https://ec.europa.eu/eurostat/statistics-explained/index.php?title=Passenger_cars_in_the_EU#Overview:_car_fleet_grows.2C_alternative_fuels_still_have_a_low_share}
  (accessed October 14, 2021).

\bibitem{frascaria2020algorithms}
Dario Frascaria and Neil Olver.
\newblock Algorithms for flows over time with scheduling costs.
\newblock In {\em 21st International Conference on Integer Programming and
  Combinatorial Optimization, IPCO 2020}, pages 130--143. Springer, 2020.

\bibitem{gale59earliestarrivalflow}
David Gale.
\newblock {Transient flows in networks.}
\newblock {\em Michigan Mathematical Journal}, 6(1):59 -- 63, 1959.

\bibitem{graf2020dynamic}
Lukas Graf, Tobias Harks, and Leon Sering.
\newblock Dynamic flows with adaptive route choice.
\newblock {\em Mathematical Programming}, 183(1):309--335, 2020.

\bibitem{harks2021stackelberg}
Tobias Harks and Anja Schedel.
\newblock Stackelberg pricing games with congestion effects.
\newblock {\em Mathematical Programming}, pages 1--37, 2021.

\bibitem{harks2019toll}
Tobias Harks, Marc Schr{\"o}der, and Dries Vermeulen.
\newblock Toll caps in privatized road networks.
\newblock {\em European Journal of Operational Research}, 276(3):947--956,
  2019.

\bibitem{matsim}
Andreas Horni, Kai Nagel, and Kay Axhausen, editors.
\newblock {\em The Multi-Agent Transport Simulation MATSim}.
\newblock Ubiquity Press, London, 2016.

\bibitem{johari2010investment}
Ramesh Johari, Gabriel~Y Weintraub, and Benjamin Van~Roy.
\newblock Investment and market structure in industries with congestion.
\newblock {\em Operations Research}, 58(5):1303--1317, 2010.

\bibitem{koch2011nash}
Ronald Koch and Martin Skutella.
\newblock Nash equilibria and the price of anarchy for flows over time.
\newblock {\em Theory of Computing Systems}, 49(1):71--97, 2011.

\bibitem{agt2007}
Noam Nisan, Tim Roughgarden, Eva Tardos, and Vijay~V Vazirani.
\newblock {\em Algorithmic Game Theory}.
\newblock Cambridge, 2007.

\bibitem{ozdaglar2008price}
Asuman Ozdaglar.
\newblock Price competition with elastic traffic.
\newblock {\em Networks: An International Journal}, 52(3):141--155, 2008.

\bibitem{pinedo2012scheduling}
Michael Pinedo.
\newblock {\em Scheduling}.
\newblock Springer, fifth edition, 2016.

\bibitem{sering2019nash}
Leon Sering and Laura {Vargas Koch}.
\newblock Nash flows over time with spillback.
\newblock In {\em Proceedings of the Thirtieth Annual ACM-SIAM Symposium on
  Discrete Algorithms}, pages 935--945. SIAM, 2019.

\bibitem{vickrey1969congestion}
William~S Vickrey.
\newblock Congestion theory and transport investment.
\newblock {\em The American Economic Review}, 59(2):251--260, 1969.

\bibitem{von2013optimal}
Philipp von Falkenhausen and Tobias Harks.
\newblock Optimal cost sharing for resource selection games.
\newblock {\em Mathematics of Operations Research}, 38(1):184--208, 2013.

\bibitem{ziemke2021flows}
Theresa Ziemke, Leon Sering, Laura {Vargas Koch}, Max Zimmer, Kai Nagel, and
  Martin Skutella.
\newblock Flows over time as continuous limits of packet-based network
  simulations.
\newblock {\em Transportation Research Procedia}, 52:123--130, 2021.

\end{thebibliography}

\end{document}